\documentclass[12pnt,aps,pra,twocolumn,superscriptaddress]{revtex4-1}

\usepackage[utf8]{inputenc}
\usepackage{graphicx}  
\usepackage{dcolumn}   
\usepackage{bm}        
\usepackage{amssymb}
\usepackage{physics}   
\usepackage{mathtools}
\usepackage{esvect}
\usepackage[normalem]{ulem}

\usepackage{amsthm}
\usepackage{ragged2e}
\usepackage{verbatim}
\PassOptionsToPackage{hyphens}{url}
\usepackage[colorlinks=true,linkcolor=blue,citecolor=red,plainpages=false,pdfpagelabels]{hyperref}
\usepackage{url}
\usepackage{color,xcolor,colortbl}
\usepackage{multirow}
\usepackage{boxedminipage}
\usepackage{changepage} 
\usepackage{bbm}
\usepackage{lipsum}
\usepackage{array}
\usepackage{etoolbox}
\usepackage[caption=false]{subfig}
\usepackage{appendix}
\usepackage{float}
\usepackage{dsfont}

\usepackage[mathscr]{euscript}

\def\Tr{\operatorname{Tr}}

\def\>{\rangle}
\def\<{\langle}

\def\id{\mathds{1}}

\newcommand{\XX}{\mathcal{X}}
\newcommand{\YY}{\mathcal{Y}}

\newcommand{\1}{I}

\newcommand{\Lrm}{\mathrm{L}}
\newcommand{\C}{\mathbb{C}}

\def\({\left(}
\def\){\right)}
\def\[{\left[}
\def\]{\right]}

\newtheorem{theorem}{Theorem}

\newtheorem{corollary}{Corollary}

\newtheorem{definition}{Definition}

\newtheorem{lemma}{Lemma}

\newtheorem{proposition}{Proposition}

\begin{document}
\title{Limitations on quantum key repeaters for all key correlated states}
\author{Leonard Sikorski}
\affiliation{Institute of Informatics, National Quantum Information Centre, Faculty of Mathematics,
Physics and Informatics, University of Gdańsk, Wita Stwosza 57, 80-308 Gdańsk, Poland}
\author{Łukasz Pawela}
\affiliation{ Institute of Theoretical and Applied Informatics, Polish Academy of Sciences, Bałtycka 5, 44-100 Gliwice, Poland}
\author{Karol Horodecki}
\affiliation{Institute of Informatics, National Quantum Information Centre, Faculty of Mathematics,
Physics and Informatics, University of Gdańsk, Wita Stwosza 57, 80-308 Gdańsk, Poland}
\affiliation{International Centre for Theory of Quantum Technologies (ICTQT), University of Gdańsk,
80-308 Gdańsk, Poland}
\affiliation{School of Electrical and Computer Engineering, Cornell University, Ithaca, New York 14850, USA}

\date{\today}

\begin{abstract}
Quantum key repeater is the backbone of the future Quantum Internet. It is an open problem to determine, for an arbitrary mixed bipartite state shared between the stations of a quantum key repeater, how much key can be generated between its two end-nodes.
We place a novel bound on the quantum key repeater rate,
which uses the relative entropy distance from, in general,
{\it entangled} quantum states. It allows us to generalize bounds on key repeaters of M. Christandl and R. Ferrara [Phys. Rev. Lett. 119, 220506]. As in the latter article, we consider a scenario used for measurement-device-independent quantum cryptography. The derived bound, although not tighter, holds for a more general class of states, thereby avoiding the NP-hard separability problem. We show that the repeated key of the broad class of key correlated states can exceed twice the one-way distillable entanglement by at most the max-relative entropy of entanglement of its attacked version. We also provide a non-trivial upper bound on the amount of private randomness of a generic independent bit - a state containing one bit of ideal private randomness.  
\end{abstract}

\maketitle
\section{Introduction}
The Quantum Internet (QI) for secure quantum communication is one of the most welcome applications of quantum information theory. In the future, Quantum Internet qubits rather than bits will be sent. However, both quantum and classical communication are subject to notorious noise in the communication channel. This is because the channel is often implemented using an optical fiber with a high attenuation parameter. To overcome this, in the "classical" Internet we are used to, the signal in transit between sender and receiver is amplified several times by intermediate stations through copying. This solution does not work in the case of the Quantum Internet, which transmits qubits, as the signal they carry cannot be amplified on the way due to the quantum no-cloning theorem \cite{Wootters1982,Dieks1982}. The seminal idea of quantum repeaters \cite{repeaters} resolves the no-cloning obstacle.

 A quantum repeater is a physical realization of a protocol that enables two distant stations to share a secure key that can be used for one-time pad encryption. There has been tremendous effort invested in building such a facility, which is now considered the backbone of future QI \cite{Wehner2018}. The most basic quantum repeater, according to the original idea of \cite{repeaters}, consists of an intermediate station $C=C_1C_2$ between two distant stations $A$ and $B$. The stations $A$ and $C_1$ receive an entangled, partially secure quantum state from a source, as do the stations $B$ and $C_2$. The role of intermediate states is to swap the security content, thereby enabling $A$ and $B$ to share a maximally entangled state. In theory, due to the seminal idea of \cite{repeaters} (for other architectures, see \cite{Muralidharan2016} and references therein), this can be achieved by altering two phases: (i) that of entanglement distillation, that is, making the entangled states less noisy i.e., as close as possible to the singlet~state 
\begin{equation}
    |\psi_+\>={1\over \sqrt{2}}(|00\>+|11\>).
\end{equation}
In the simplest setup mentioned above, this step enhances entanglement between $A$ and $C_1$ station
and $C_2$ and $B$ respectively. The second phase (ii) is entanglement swapping \cite{ukowski1993}, i.e., teleporting one subsystem $C_1$ of the state distilled in point (i) on $AC_1$ using the other singlet $C_2B$ to obtain a singlet between $A$ and $B$. Increasing the number of stations results, theoretically, in arbitrarily distant quantum connections realized by means of the singlet state.

One of the main reasons for building the Quantum Internet is the need for secure communication. Distributing pure entangled states is only one of the ways to achieve such communication.
In general, the underlying shared output state of such a secure key distribution protocol is described by the so-called private state - a state containing privacy in ideal form  \cite{pptkey,keyhuge}.
A private state containing (at least) one bit of ideal key is a bipartite state on systems $ABA'B'$ that takes the following form when purified to the system of eavesdropper $E$:
\begin{align}
    {1\over {\sqrt 2}}&\left(|00\>_{AB}\otimes (U_0^{A'B'}\otimes \mathrm{I}_E) |\phi\>_{A'B'E} + \right.\nonumber\\
    &\left.|11\>_{AB}\otimes (U_{1}^{A'B'}\otimes \mathrm{I}_E) |\phi\>_{A'B'E}\right),
    \label{eq:first_pbit}
\end{align}
where $U_{i}$ for $i\in\{0,1\}$ are unitary transformations that act on a system $A'B'$ only and $|\phi\>_{A'B'E}$ is arbitrary state of $A'B'E$. One can see the correspondence of the unitary rotations $U_i$ to the complex phases in a usual maximally entangled state of the form $|\psi_{max}\>= {1\over \sqrt{2}}(e^{\theta_0 i}|00\> + e^{\theta_1 i} |11\>)$. The secure key is obtained by measuring systems $A$ and $B$ of the private state in the computational basis. This is why the system $AB$ is called the key part. The system $A'B'$ plays a passive role of shielding system $AB$ from the system $E$ of Eve, and it is therefore called a 'shield'.

For this reason, the idea of quantum repeaters has been generalized in \cite{BCHW15} to the notion of {\it quantum key repeater}. The quantum key repeater uses the same station setup $AC_1C_2B$. However, the task of the protocol is to output a private state between stations $A$ and $B$ (which need not be maximally entangled in general). The goal is to distribute the secure key not in the form of pure entangled states (the singlet) but in the form of general mixed quantum states that nonetheless contain ideal security. 

There have been established bounds on the performance of protocols for distributing the key, secure against quantum adversaries, using quantum key repeaters. The first obtained works for the so-called states with positive partial transposition \cite{BCHW15}. Later, it was partially generalized to the case of the states that resemble private states in their structure \cite{CF17}. Specifically, in \cite{CF17}, a significant bound has been shown on a variant of the key repeater rate $R$ for the case of two special so-called {\it key correlated states} $\rho_{AC_1}=\rho_{C_2B}\equiv\rho$. To recall this bound below, we also recall the notion of key-attacked states of the key correlated states. Namely, the key correlated states, like private states, have two distinguished subsystems $\rho_{AC_1}\equiv \rho_{AA'C_1C_1'}$ where, with an abuse of notation, we identify system $A $ with a key part of a system $A$ and shielding system $A'$ (similarly for $C_1$). Now, when the key part systems $AC_1$ got measured in computational basis $\{|ij\>_{AC_1}\}$, the state becomes the key-attacked state of the key correlated state, and is denoted as $\hat{\rho}_{AA'C_1C_1'}$ (see Definition \ref{def:kas}), and analogously for $\hat{\rho}_{C_2C_2'BB'})$. With these notions, the bound of \cite{CF17} applies to those key correlated states, the key attacked-state of which are separable, i.e., totally insecure:
\begin{equation}
    R^{C_1C_2\rightarrow A:B}(\rho_{AC_1},\rho_{C_2B})\leq 2E_D^{C_1\rightarrow A}(\rho_{AC_1}).
    \label{eq:mian_CF_bound}
\end{equation}
Here $E_D^{\rightarrow}$ is the {\it one-way distillable entanglement} \cite{Bennett1996}, that is, the amount of entanglement in the form of maximally entangled states that can be distilled from $\rho_{CA}$ via local quantum operations and one-way (from Charlie to Alice) classical communication operations. The $C_1C_2\rightarrow A:B$ denotes the fact that classical communication in the repeater's protocol takes direction from $C_1C_2$ to $AB$ and later arbitrarily between $A$ and $B$ stations. Note that this is precisely the scenario of public communication known in measurement-device-independent cryptography \cite{lo2012measurement, braunstein2012side}, hence the presented bound applies to this relevant quantum-cryptographic scenario (see \cite{christandl2017relative, azuma2016fundamental, pirandola2019end, das2021universal}) in this context.

In what follows, we would like to drop the assumption that the states $\rho_{AC_1}$ and $\rho_{C_2B}$ get separable after an attack on their key parts by measurement in the computational basis. That is, more precisely, that their key-attacked states are separable. The reason is that checking separability for a given quantum state is a computationally hard problem (precisely NP-hard). In that sense, we broaden the class of states for which an extension of the bound of \cite{Christandl2017} holds. We do so by relaxing it slightly so that the resulting bound is not tighter than the previous one but is naturally less stringent.

In the second part of this manuscript, we consider the scenario of private randomness distillation from a bipartite state
by two honest parties \cite{YHW19}. In the latter scenario, the parties share n copies of a bipartite state $\rho_{AB}$ and distill private randomness
in the form of the so-called {\it independent}  states. These are states from which, after local measurements, one can obtain ideal, uniformly random bits for two parties that are private, i.e., decorrelated from the eavesdropper's system. The operations that the parties
can perform in this resource theory are (i) local unitary operations and (ii) sending system via a dephasing channel. In \cite{YHW19}, cases in which communication is not allowed and maximally mixed states are locally accessible to the parties (or are disallowed) were also considered. In all four cases, the achievable rate regions were provided, and they were tight in most cases.

Here, we study the amount of private randomness in generic local independent states. The latter states are bipartite states $\alpha_{AA'B'}$ of the form
\begin{equation}
    \alpha_{AA'B'}= \sum_{i,j=0}^{1}{1\over 2}|i\>\<j|_A\otimes U_i \sigma_{A'B'} U_j^{\dagger}
\end{equation}
possessing one bit of ideal randomness accessible via the measurement in the computational basis on system $A$. The other party, holding system $B'$ is honest in this scenario. We consider a random local independent state and prove that, given large enough systems $d_s$ of $A'$ and $B'$, in the above-mentioned scenarios, the localisable private randomness $R_A$ at the system of Alice is bounded as follows:
\begin{equation}
    R_A(\alpha_{AA'B'})\leq 1 + \frac{1}{2\ln 2}
\end{equation}
and this rate is achievable for asymptotically growing dimension of the shielding system at $A'$, $d_s$, in all four scenarios.

The remainder of the manuscript is organized as follows. Section \ref{sec:technical} is devoted to technical facts and definitions used throughout the rest of the manuscript.
Section \ref{sec:main} provides the bounds
on key repeater rate in terms of distillable entanglement and the R\'enyi relative entropy of entanglement. Section \ref{sec:key_bound} presents a simple bound on the distillable key of a random private state. In Section \ref{sec:ibits}, we also provide bounds on private randomness for a generic independent state. We finalize the manuscript with a short discussion in Section \ref{sec:discussion}. 

\section{Main results}
In this manuscript, we develop a relaxed bound on quantum key repeater rate, which holds for the key-correlated states without the assumption that the $\hat{\rho}$ is separable. This relaxation is essential, as checking separability is, in general, an NP-hard problem \cite{Gurvits}. 
We provide a bound on the key repeater rate using the relative entropy distance to {\it arbitrary} states rather than to separable ones. The freedom in choosing a state in the relative entropy distance is, however, compensated by another relative-entropy-based term (the sandwich R\'enyi relative entropy of entanglement $\tilde{E}_\alpha$ \cite{Hayashi2017})  and a pre-factor $\alpha$
 from the interval $(1,+\infty)$. It reads the following upperboundbound on $R^{C_1C_2\rightarrow A:B}(\rho_{AC_1},\rho_{C_2B})$: 
\begin{align}
    \inf_{\alpha \in (1,\infty)}\left[{\alpha \over \alpha -1}E_D^{C_1C_2\rightarrow A}(\rho_{AC_1}\otimes \rho_{C_2B})  \right. \nonumber\\
    \left.\qquad + \min\{{\tilde E}_{\alpha}(\hat{\rho}_{AC_1}), {\tilde E}_{\alpha}(\hat{\rho}_{C_2B})\}\right],
    \label{eq:main_bound}
\end{align}
where $\hat{\rho}_{AC_1}$ and $\hat{\rho}_{C_2B}$ are the key-attacked states of the key correlated states $\rho_{AC_1}$ and $\rho_{C_2B}$ respectively (see Definition \ref{def:kas}).
Although the above bound is not tighter than the known one, it holds for a more general class of states: we do not demand the key-attacked states $\hat{\rho}$ to be separable. We note here that for the case when $\hat{\rho}$ is separable, we recover the bound (\ref{eq:mian_CF_bound}) by taking the limit of $\alpha \rightarrow \infty$. In this case, $\tilde{E}_\alpha$ tends to the max-relative entropy \cite{Datta}, which is zero for separable states, while the factor $\frac{\alpha}{\alpha -1}$ goes to $1$. In proving the above bound, we base on the strong converse bound on private key recently shown in \cite{WTB2017,Das2020} (see also \cite{Christandl2017,DBWH20} in this context). In particular, the bound given above implies, for an arbitrary key correlated state
\begin{equation}
R^{C_1C_2\rightarrow A:B}(\rho_{AC_1},\rho_{C_2B})\leq 2E_D^{C_1\rightarrow A}(\rho_{AC_1}) + E_{max}(\hat{\rho}_{AC_1}),
\label{eq:more_relaxed}
\end{equation}
where the $E_{max}$ is the max-relative entropy of entanglement \cite{Datta,Datta2}.

We further show an upper bound on the distillable key of a {\it random} private bits \cite{pptkey,keyhuge}. These states contain ideal keys for one-time pad encryption, secure against a quantum adversary who can hold their purifying system. The importance of this class of states stems from the fact that, as recently shown, they can be used in the so-called {\it hybrid quantum key repeaters} to enhance the security of the Quantum Internet \cite{Sakarya2020}.
The first study on generic private states was done in \cite{Christandl2020}. We take a different randomization procedure than the one utilized there. We base it on the fact that every private bit can be represented by, in general, not normal, operator $X$ \cite{keyhuge}. The latter, in turn, can be represented as $X=U\sigma$ for some unitary transformation and a state $\sigma$. 
We utilize techniques known from standard random matrix theory (see Appendix) to draw a random $X$ and upper bound the mutual information of the latter state, half of which upper bounds the distillable key \cite{squashed}. Based on this technique via the bound of Eq. (\ref{eq:more_relaxed}) we show, that a {\it randomly} chosen private bit $\gamma_{rand}$, satisfies:
\begin{align}
    &R^{C_1C_2\rightarrow A:B}(\gamma_{rand},\gamma_{rand})\leq
    2E_D^{\rightarrow}(\gamma_{rand}) + 1.
    \label{eq:crude}
\end{align}

Recall that a private bit has two distinguished systems. That of the key part, from which the von-Neumann measurement can directly obtain the key in the computational basis, and the system of shield, the role of which is to protect the key part from the environment. In this language, the state $U_0\sigma U_0^{\dagger}$ and $U_1\sigma U_1^{\dagger}$ are the ones appearing on its shield part, given the key of value $0$ (or respectively $1$) is observed on its key part, and are called {\it the conditional states}. Any private bit $\gamma_2$ has distillable key at least equal to $1 + \frac12 K_D(U_0\sigma U_0^{\dagger} \otimes U_1\sigma U_1^{\dagger})$ as it is shown in \cite{HCRS18}. 
From what we prove here, it turns out that the amount of key that can be drawn from the conditional states is bounded by a constant {\it irrespective of the dimension of the shield part}. More precisely, 
we derive a bound on $R^{C_1C_2\rightarrow A:B}(\gamma_{rand},\gamma_{rand})$ that is independent from Eq. (\ref{eq:crude}), by showing that a random private bit $\gamma_{rand}$ satisfies
\begin{equation}
    K_D(\gamma_{rand}) \leq 1 + \frac{1}{4\ln 2}\approx 1.360674.
\end{equation}

\section{Technical preliminaries}
\label{sec:technical}
In this section, we recall certain facts and definitions and fix notation.
A {\it private state} is a state of the form
\begin{align}
    &\gamma_{ABA'B'}=\frac{1}{d_k}\sum_{ij=0}^{d_k-1}|ii\>\<jj|_{AB}\otimes U_i \rho_{A'B'}U_j^{\dagger}\equiv\nonumber\\
    &\tau |\psi_+\>\<\psi_+|_{AB}\otimes \rho_{A'B'} \tau^{\dagger}
    \label{eq:private-state}
\end{align}
where $U_i$ are unitary transformations,
and $\rho_{A'B'}$ is an arbitrary state on system $A'B'$ of dimension $d_s\otimes d_s$. A controlled unitary transformation $\tau := \sum_{i}|ii\>\<ii|\otimes U_{i} + \sum_{i\neq j}|ij\>\<ij|\otimes \id$ where $U_{i}$ are unitary transformations, is called a {\it twisting} and $|\psi_+\>\<\psi_+|=\sum_{i,j=0}^{d_k-1}\frac{1}{d_k}|ii\>\<jj|$. In the case of $d_k=2$, it can be obtained from Eq. (\ref{eq:first_pbit}) by tracing out system $E$. System $AB$ is called the key part, and system $A'B'$ is called the shield. We denote by $\gamma^k$ a private state with $2^k \otimes 2^k$ key part, i.e., containing {\it at least} $k$ key bits. In the case of $k=1$, the private state is called a {\it private bit}, or {\it pbit}, while for $k>1$, it is called a {\it pdit} (with $d_k = 2^k$). A state $\rho$ is an $\epsilon$-approximate private state $\gamma_{d_k}$ if there exists $\epsilon>0$ such that $||\rho - \gamma_{d_k}||_{tr}\leq \epsilon$.

The notion of private states allowed the problem of distillation of secret key described by the so-called Local Operations and Public Communication  \cite{Devetak2005} in a tripartite scenario with Alice Bob and an eavesdropper Eve, to be described as a problem of distillation of private states by Local Operations and Classical Communication (i.e. a composition of local quantum operations and classical communication in both directions) in the worst case scenario where Eve holds a purification of the state shared by Alice and Bob. We recall the obtained definition of the distillable key below.
 \cite{keyhuge,pptkey}
\begin{align}
    K_D(\rho):=\inf_{\epsilon >0}\limsup_{n\rightarrow \infty} \sup_{{\cal P}\in LOCC} 
    \{\frac{k}{n} : {\cal P}(\rho^{\otimes n})\approx_\epsilon \gamma^k \},
\end{align}
where $\rho \approx_\epsilon \sigma$ denotes $||\rho -\sigma||_{tr} \leq \epsilon$, i.e. closeness is the trace-norm distance by $\epsilon$,  with $||X||:= Tr\sqrt{XX^{\dagger}}$ and ${\cal P}$ is a completely positive trace-preserving map from the set of LOCC operations.

 A technical, but important, role in our considerations is played by the so-called {\it key attacked private state}, denoted as $\hat{\gamma}$. This is a private state that got measured on its key part and takes the form
\begin{equation}
    \hat{\gamma}:=\frac{1}{d_k}\sum_{i=0}^{d_k-1}|ii\>\<ii|_{AB}\otimes U_i \rho_{A'B'}U_i^{\dagger}.
\end{equation}
We note here, that $\hat{\gamma}$ is separable iff $U_i\rho U_i^{\dagger}$ are separable for $i\in \{0,\ldots, d_k-1\}$. We aim at generalizing the results of \cite{CF17} to the case when we do not know if a certain key attacked state is separable.

Below, we recall the main definitions and results from \cite{CF17}.
There, the notion of {\it key correlated states} is introduced. This class of states is the one for which the bound for quantum key repeaters is given in \cite{CF17}. 
Denoting by ${\cal Z}_{d_k}$ the Weyl operator of dimension $d_k$ of the form ${\cal Z}_{d_k} = \sum_{l=0}^{d_k-1} e^{i2\pi l/d_k } |l\>\<l| $ we obtain the generalized Bell states as 
$\phi_j := {\cal Z}_{d_k}^j\otimes \id_{B}(\phi)$ where
$\phi:= {1\over \sqrt{d_k}} \sum_i |i\>_A|i\>_B$. Then,
the {\it key correlated } state takes form:
\begin{equation}
       \rho_{ABA'B'}^{key-cor}\coloneqq \sum_{\mu,\nu} |\phi_\mu\>\<\phi_\nu|_{AB}\otimes M^{(\mu,\nu)}_{A'B'},
    \label{eq:key_correlated}
\end{equation}
where $M^{(\mu,\nu)}$ are $d_s\times d_s$ matrices on $A'B'$.

The notion of the key attacked state of a private state naturally generalizes to the case of a key attacked state of the key correlated result. It can be obtained as follows:
\begin{definition}
\label{def:kas}
The key attacked state $\hat{\rho}_{ABA'B'}$ of a key correlated state $\rho^{key-cor}_{ABA'B'}$ defined in equation \eqref{eq:key_correlated} is the state
obtained from $\rho^{key-cor}_{ABA'B'}$ by applying local von-Neumann measurements on systems $A$ and $B$ in the computational basis.
It takes the following form: 
\begin{equation}
\hat{\rho}_{ABA'B'}\coloneqq\sum_{i} p_i |ii\>\<ii|_{AB}\otimes \sigma_{A'B'}^{(i)}
\end{equation}
where $\sigma_{A'B'}^{(i)}$ are some states on the $A'B'$ system such that $\sum_i p_i \sigma_{A'B'}^{(i)}=Tr_{AB}\rho_{ABA'B'}$ and $\{p_i\}$ forms a probability distribution.
\end{definition}

One of the main results of \cite{CF17} connects the problem of distinguishability of a key correlated state from its key attacked version.
It states that the one-way distillable entanglement of the key correlated state $\rho$ is quantified by means of the so-called locally measured-relative entropy "distance" between the $\rho$ and its key attacked version when the latter state is separable. By locally measured relative entropy distance between two states we mean the relative entropy of two states measured both by a quantum completely positive, trace preserving map ${\cal M}$ (not necessarily a POVM as it was defined in \cite{Piani}) on Alice's side:
\begin{align}
    &\sup_{{\cal M} \in LO_A} D({\cal M}(\rho)||{\cal M}(\sigma)):=\nonumber\\
    &\sup_{{\cal M}}D(
    ({\cal M}\otimes \mathrm{I}_B) \rho ||({\cal M}\otimes \mathrm{I}_B) \sigma)
\end{align}
where $D(\rho||\sigma)= \mathrm{Tr}\rho\log\rho - {\mathrm Tr}\rho\log \sigma$.
In what follows, we need a regularized
version of it, which reads:
\begin{align}
    &\sup_{{\cal M}\in LO_A} D^{\infty}(\rho||\sigma):=\nonumber\\&\lim_{n\rightarrow \infty} \frac1n \sup_{{\cal M}\in LO_A}D({\cal M}(\rho^{\otimes n})||{\cal M}(\sigma^{\otimes n})).
    \label{eq:reg_mes_relent}
\end{align}
where in the above ${\cal M}$ acts naturally on all the $n$ subystems $A$ of $\rho$ (and $\sigma$ respectively).

In what follows, we will need to recall the definition of one-way distillable entanglement. It is the maximal number of approximate singlet states that can be produced from $n$ copies of the input state $\rho$ via LOCC operations that use only communication from $A$ to $B$ (denoted as $LOCC(A\rightarrow B)$), in the asymptotic limit of large $n$ and vanishing error of approximation.
Formally it reads:
\begin{align}
&E_D^{A\rightarrow B}(\rho):=\inf_{\delta >0}\limsup_{n\rightarrow \infty}\sup_{\Lambda \in LOCC(A\rightarrow B)}\nonumber\\ &\{E:\Lambda(\rho^{\otimes n})\approx_{\delta} |\psi_+\>\<\psi_+|^{\otimes nE}\}
\end{align}

We are ready to invoke the above-mentioned result of \cite{CF17}.
\begin{theorem}[cf. Theorem 2 in \cite{CF17}]
For any key correlated state $\rho$, and its key attacked state $\hat{\sigma}$, it holds 
\begin{align}
    &E_D^{\rightarrow}(\rho)\geq D_A(\rho||\hat{\rho}):= \sup_{{\cal M}\in LO_A} D({\cal M}(\rho)||{\cal M}(\hat{\rho})),\nonumber\\
&E_D^{\rightarrow}(\rho)\geq D_A^{\infty}(\rho||\hat{\rho}):= \sup_{{\cal M}\in LO_A} D({\cal M}(\rho^{\otimes n})||{\cal M}(\hat{\rho}^{\otimes n})).
\end{align}
If $\hat{\rho}$ is separable then:
\begin{equation}
    E_D^{\rightarrow}(\rho)=D^{\infty}_A(\rho||\hat{\rho}).
\end{equation}
\label{thm:CF_dist}
\end{theorem}

A specific variant of the key repeater rate plays a central role in our considerations. 
The {\it one-way key repeater rate}
is defined in \cite{CF17}, as  follows, where the protocols consist of Local Operations and Classical Communication LOCC operations  that are a composition of  (i) local quantum operation on system $C_1C_2\equiv C $ followed by one-way  classical communication, i.e., the one  from $C_1C_2\equiv {\bf C}$ to
$AB$ followed by (ii) an arbitrary LOCC operation between $A$ and  $B$. We denote this composition as $LOCC(C_1C_2 \rightarrow A \leftrightarrow B)$.
\begin{definition} (cf. Appendix G in \cite{CF17})
Consider any two bipartite states $\rho_{AC_1}$ of systems $AC_1$ and $\rho'_{C_2B}$ of systems $C_2B$ (not necessarily equal), called further input states. A protocol ${\cal P}$ consisting of $LOCC(C_1C_2\rightarrow A\leftrightarrow B)$ operations which transform $n$ copies of the input states into $\epsilon$ approximation (in trace norm distance) of a private state $\gamma^k$ shared between $A$ and $B$ has the key rate $\frac{k}{n}$. The maximal value of the key rates over such protocols in the asymptotic limit of large $n$ and vanishing error $\epsilon$ is called a one-way key repeater rate of $\rho_{AC_1}\otimes\rho_{C_2B}'$, and reads
\begin{align}
    &R^{C_1C_2\rightarrow AB}(\rho_{AC_1},\rho'_{C_2B})\coloneqq\lim_{\epsilon \rightarrow 0}\lim_{n\rightarrow\infty} \sup_{{\cal P}\in LOCC(C_1C_2\rightarrow A\leftrightarrow B)}\nonumber\\
    &\{ \frac{k}{n}: {\cal P}((\rho_{AC_1}\otimes \rho_{C_2B}')^{\otimes n})\approx_\epsilon \gamma_{AB}^k\}.
\end{align}
\end{definition}

Following \cite{CF17}, we will use the locally measured on system ${\bf C}\equiv C_1 C_2$, regularized relative entropy applied to the scenario of key repeaters:
\begin{align}
    &D_{\bf{C}}^{\infty}(\rho_{AC_1}\otimes \rho'_{C_2B}||\rho_{AC_1C_2B}'')\coloneqq \lim_{n \rightarrow \infty} {1\over n}\nonumber \\& 
    \sup_{{\cal M}_{\bf C}\in LO_{C_1C_2}} D({\cal M}_{\bf C} ({\rho_{AC_1}\otimes
    \rho'_{C_2B}}^{\otimes n})||{\cal M}_{\bf C}( {\rho''_{AC_1C_2B}}^{\otimes n})),
\end{align}
which is nothing but the locally measured relative entropy with $A$ in Eq. (\ref{eq:reg_mes_relent}) identified with
the system ${\bf C}=C_1C_2$ and $B$ with $AB$ respectively, for
the input state $\rho_{AC_1}\otimes\rho_{C_2B}$. Let us emphasize here, that ${\cal M}_{\bf C}$ acts locally on all $n$ subsystems $C_1C_2$ of  the state $(\rho_{AC_1}\otimes \rho_{C_2B})^{\otimes n}$ and  analogously on all $n$ subsystems $C_1C_2$ of the state ${\rho''_{AC_1C_2B}}^{\otimes n}$.

As it was further noted in \cite{CF17}, there holds a bound obtained from the Theorem 4 of \cite{BCHW15}:
\begin{align}
    R^{C_1C_2\rightarrow AB}(\rho_{AC_1},\rho'_{C_2B})\leq D_{\bf{C}}^{\infty}(\rho_{AC_1}\otimes \rho'_{C_2B}||\sigma)
\end{align}
for any $\sigma$ separable in $AC_1C_2:B$ or $A:C_1C_2B$ cut. 
Then for $\rho,\rho'$ being key correlated states, and by choosing $\sigma = \hat{\rho}\otimes \hat{\rho'}$ (i.e. assuming that the key attacked the correlated state of at least one of the two states is separable) in \cite{CF17} they have obtained from Theorem \ref{thm:CF_dist} invoked above 
\begin{align}
    R^{C_1C_2\rightarrow AB}(\rho_{AC_1},\rho'_{C_2B})\leq E_D^{C_1C_2\rightarrow AB}(\rho_{AC_1}\otimes \rho'_{C_2B})
\end{align}

The main contribution of this manuscript amounts to obtaining a bound for which neither of the states $\hat{\rho}$ and $\hat{\rho'}$ need to be separable.

Further, the following quantities are relevant for the proof technique. We will need the notion of {\it $\epsilon$-hypothesis-testing divergence} \cite{Buscemi,LR}, which is defined as
\begin{equation}
    D^{\epsilon}_h(\rho||\sigma):=-\log_2 \inf_{{\mathrm P}: 0\leq {\mathrm P}\leq \id} \{\Tr[{\mathrm P} \sigma]: \Tr[{\mathrm P} \rho]\geq 1-\epsilon\},
    \label{eq:d_epsilon}
\end{equation}
and the {\it sandwiched R\'enyi relative entropy} \cite{WWY14,MDSFT}:
\begin{equation}
    \tilde{D}_\alpha(\rho||\sigma):={1\over\alpha -1}\log_2\Tr[\left(\sigma^{1-\alpha \over 2\alpha}\rho\sigma^{1-\alpha \over 2\alpha}\right)^\alpha].
    \label{eq:Renyi}
\end{equation}
In \cite{Hayashi2017}, based on the above quantity, the {\it sandwiched R\'enyi relative entropy of entanglement} is defined
as follows:
\begin{equation}
    {\tilde E}_{\alpha}(\rho):= \inf_{\sigma \in SEP} \tilde{D}_\alpha(\rho||\sigma),
\end{equation}
where $SEP$ denotes the set of separable states \cite{4H}.
As it is shown in \cite{Christandl2017} (and independently in \cite{Zhu2017}), for $\alpha \rightarrow \infty$ the above quantity takes as
a limit the {\it max relative entropy of entanglement} \cite{Datta,Datta2}, where the latter quantity is defined as:
\begin{equation}
    E_{max}(\rho) = \inf_{\sigma \in SEP} \inf\{\lambda \in{\mathrm R}: \rho \leq 2^{\lambda}\sigma\}.
\end{equation}

As the second part of this work heavily relies on {\it random} quantum objects, we provide a brief introduction to random matrices in Appendix~\ref{app:random}.

Finally, we would like to recall the results of \cite{YHW19} on {\it private randomness} distillation, as our result also applies to this resource. There, the notion of {\it independent states} was introduced. These states take the form of 
twisted coherence, just like the private state is twisted entanglement. 
\begin{definition} The independent state is defined as follows
\begin{align}
& \alpha_{d_A,d_B} :=\\
& \qquad {1\over d_Ad_B} \sum_{i,j=0}^{d_A-1}\sum_{k,l=0}^{d_B-1} |i\>\<j|_A\otimes |k\>\<l|_B\otimes U_{ik}\sigma_{A'B'}U_{jl}^{\dagger} \equiv \nonumber\\
& \qquad \qquad \tau |+_{d_A}\>\<+_{d_B}|_A\otimes |+_{d_B}\>\<+_{d_B}|_B\otimes \sigma_{A'B'} \tau^\dagger\nonumber
\end{align}
where the unitary transformation $\tau =\sum_{ij}|ij\>\<ij|\otimes U_{ij}$ is the 
twisting, $\sigma_{A'B'}$ is an arbitrary
state on $A'B'$ system and $|+_d\> = 1/\sqrt{d}\sum_{i=0}^{d-1}|i\>$.
\end{definition}
The independent state has the property that, when measured on the system $AB$, yields
$\log_2 d_A$ bits of ideally private randomness
for Alice and $\log_2 d_B$ for Bob. 
The randomness is ideally private since the construction of the state assures that the outcomes of measurements on the key part system $AB$ are decorrelated from the purifying system $E$.
We also recall the scenario of distributed randomness distillation of \cite{YHW19}.
In the latter scenario, two parties are trying to distill locally private randomness, which is independent for each party and decoupled from the environment's system. The honest parties share a dephasing channel via which they can communicate ({\it free communication}). One can also consider the case in which the honest parties have local access to an unlimited number of maximally mixed states ({\it free noise}). The parties, $A$ and $B$, distill private randomness at rates $R_A$ and $R_B$, respectively, in the form of independent states. We recall below the main results of \cite{YHW19}:

{\theorem [Cf. \cite{YHW19}]
\label{thm:priv_rand} 
The achievable rate regions of $\rho_{AB}$ are:
\begin{enumerate}
\item for no communication and no noise, 
$  R_A       \leq \log|A| - S(A|B)_+ $, 
$  R_B       \leq \log|B| - S(B|A)_+ $, and
 $ R_A + R_B \leq R_G$,
where $[t]_+=\max\{0,t\}$;

\item for free noise but no communication, 
$R_A       \leq \log|A| - S(A|B)$, 
$R_B       \leq \log|B| - S(B|A)$, and 
$R_A + R_B \leq R_G$;

\item for free noise and free communication, 
$R_A\leq R_G$, $R_B\leq R_G$, and $R_A+R_B\leq R_G$;

\item for free communication but no noise, 
$R_A       \leq \log|AB| - \max\{S(B),S(AB)\}$, 
$R_B       \leq \log|AB| - \max\{S(A),S(AB)\}$, and
$R_A + R_B \leq R_G$,
\end{enumerate}
where $R_G$ denotes {\it global purity} and equals $\log|AB| - S(AB)$}.
Further, the rate regions in settings 1), 2), 3) are tight.

In the above $S(A|B)_+=\max \{0,S(A|B)\}$. Here $S(X)_{\rho_{XY}}$ is the von-Neumann entropy of subsystem $X$ of the state $\rho_{XY}$, $S(X|Y)_{\rho_{XY}}$ is the conditional von-Nuemann entropy while 
\begin{equation}
    I(X:Y)_{\rho_{XY}} \coloneqq S(X)_{\rho_{XY}} + S(Y)_{\rho_{XY}} - S(XY)_{\rho_{XY}}
\end{equation}
is the quantum mutual information. We will sometimes neglect the subscript $XY$ in 
the above if it is understood from the context.
Recently, it has been shown that private states are independent states \cite{HKSS20}. However, naturally, the set of independent states is strictly larger. In particular, a {\it local independent} bit is not necessarily a private state:
\begin{equation}
    \alpha_{d_A}:=\tau |+\>\<+|_A\otimes \sigma_{A'B'}\tau^{\dagger},
    \end{equation}
where $|+\>={1\over \sqrt{d}_A} \sum_{i=0}^{d_A} |i\>$.
In what follows, we will estimate the private randomness content of a generic local independent bit.

\section{Relaxed bound on one-way private key repeaters}
\label{sec:main}

In this section, we provide a relaxed bound on one-way private key repeaters. It holds for the so-called key correlated states $\rho_{key}$ \cite{CF17}.
The bound takes the form of the regularized, measured relative entropy of entanglement from {\it any} state scaled by the factor $\alpha \over \alpha -1$ increased by the $\alpha$-sandwiched relative entropy of entanglement $E_\alpha$ of the key-attacked version of $\rho_{key}$. 
Although the latter bound can not be tighter than the bound by relative entropy of entanglement \cite{pptkey,keyhuge}, it has
appealing form, as it, in a sense, computes the distance from a separable state via a 'proxy' state, which can be arbitrary.

As the main technical contribution, we will first present the lemma, which upper bounds the fidelity with a singlet of a state which up to its inner structure (being outcome of certain protocol) is {\it arbitrary}. Such a bound was known so far for the fidelity of a singlet with the so-called {\it twisted separable states} (see Lemma 7 of \cite{keyhuge}) and proved useful in providing upper bounds on distillable key (see \cite{KhatriWildeBook} and references therein). Our relaxation is based on the strong-converse bound for the private key, formulated in terms of the $\tilde{E}_{\alpha}$.

We begin by explaining the idea of a technical lemma, which we present below. It is known that any $d\otimes d$ separable state is bounded away as $1\over d$ from a maximally entangled state of local dimension $d$ in terms of fidelity \cite{keyhuge}. However, an arbitrary state can not be 
bounded arbitrarily away from the singlet state in terms of fidelity.
There must be a penalty term that quantifies how rapidly this fidelity increases as the state becomes increasingly entangled, i.e., far and far from being separable. Such a term
follows from the strong converse bound on the distillable key. Indeed, this is the essence of $\tilde{E}_\alpha$ being the strong converse bound: If one tries to distill more key than the $\tilde{E}_\alpha$ from a state, then
the fidelity with a singlet of the output state increases {\it exponentially
fast} to $1$, in terms of $\tilde{E}_\alpha$. In our case, the state compared with the singlet in terms of fidelity will be the state after the performance of the key repeater protocol, subjected to twisting and partial trace. Namely, a state $\sigma^{\otimes n}$ of systems $ABC$ after the action of a map $\Lambda\in LOCC({\bf C}\rightarrow (A\leftrightarrow B))$ which can be now of systems $ABA'B'C$, traced over $C$ and further rotated by a twisting unitary transformation $\tau$ and traced over system $A'B'$. The last two actions, i.e., rotation and twisting, serve only as a mathematical tool to check whether the state under consideration, here the protocol's output, is close to a private state. Indeed, a private state, after performing the inverse operation to twisting and the partial trace, becomes the singlet state \cite{pptkey}. We note here that the proof technique in the Lemma (\ref{lem:tr}) below bears a partial analogy to the proof of the fact that the so called second type distillable entanglement is upper bounded by the relative entropy of entanglement as shown in \cite{Vedral1998} and described in Theorem $8.7$ of \cite{Hayashi2017}.

\begin{lemma}
\label{lem:tr}
For a bipartite state of the form
$\tilde{\sigma} := \Tr_{A'B'}\tau\Tr_C\Lambda(\sigma^{\otimes n})\tau^{\dagger}$, where $\sigma_{ABC}$ is arbitrary state, $\tau$ is an arbitrary twisting, and $\Lambda \in LOCC({\bf C}\rightarrow (A\leftrightarrow B))$ there is
\begin{equation}
    \Tr \tilde{\sigma} \psi_+^{nR} \leq 2^{-n(\frac{\alpha -1}{\alpha})[R-1/n\tilde{E}_\alpha(\sigma^{\otimes n})]},
\end{equation}
where $R >0$ is a real number, and $\psi_+$ denotes the maximally entangled state, and $\alpha \in (1,\infty)$.
\end{lemma}
\begin{proof}
In what follows, we directly use the proof of the strong-converse bound for the distillable key of \cite{DBWH20}, which is based on \cite{WTB2017,Das2020,Christandl2017}.
Let us suppose $\Tr \tilde{\sigma} \psi_+^{nR}= 1-\epsilon$. Let us also choose 
a state $\hat{\sigma}$, to be a twisted separable state of the form $\hat{\sigma} = \Tr_{A'B'}\tau\Tr_C\Lambda(\sigma_{A{\bf C}:B}')\tau^{\dagger}$ with $\sigma'$ being arbitrary separable state in $(A{\bf C}:B)$ cut. Note here that by definition $\Lambda$ preserves separability in $A{\bf C}:B$ cut, and
hence in $A:B$ as well.

Since any such state has an overlap
with the singlet state $\psi^{nR}_+$ less
than $\frac{1}{2^{nR}}\equiv 1/K$ (see lemma 7 of \cite{keyhuge}), we have by taking ${\mathrm P}=|\psi_+^{nR}\>\<\psi_+^{nR}|$ definition of $D^{\epsilon}_h$ given in Eq. (\ref{eq:d_epsilon}) a bound:
\begin{equation}
    \log_2 K \leq D^{\epsilon}_h(\tilde{\sigma}||\hat{\sigma}).
\end{equation}
We then upper bound $D^{\epsilon}_h$ by $\tilde{D}_\alpha$ as follows, for every $\alpha \in (1,\infty)$ (see A8 in the Appendix of \cite{DBWH20}):
\begin{align}
    &\log_2 K \leq \tilde{D}_\alpha(\tilde{\sigma}||\hat{\sigma}) +   \frac{\alpha}{\alpha -1} \log_2(\frac{1}{1-\epsilon})  =\nonumber \\&
    \tilde{D}_\alpha(\Tr_{A'B'}\tau\Tr_C\Lambda(\sigma^{\otimes n})\tau^{\dagger} ||\Tr_{A'B'}\tau \Tr_C\Lambda(\sigma'_{A{\bf C}B} )\tau^{\dagger}) +\nonumber \\& \frac{\alpha}{\alpha -1} \log_2(\frac{1}{1-\epsilon}).
\end{align}
We can now relax $\tilde{\sigma}$ to $\Lambda(\sigma^{\otimes n})$ due to 
monotonicity of the sandwich R\'enyi relative entropy distance under a jointly applied channel, so that
\begin{equation}
    \log_2 K \leq \tilde{D}_\alpha(\Lambda (\sigma^{\otimes n})||\Lambda(\sigma')) +   \frac{\alpha}{\alpha -1} \log_2(\frac{1}{1-\epsilon}).
\end{equation}
We can also further drop the operation $\Lambda$ again using monotonicity of the relative R\'enyi divergence $\tilde{D}_\alpha$. We thus arrive at
\begin{equation}
    \log_2 K \leq \tilde{D}_\alpha(\sigma^{\otimes n}||\sigma') +   \frac{\alpha}{\alpha -1} \log_2(\frac{1}{1-\epsilon}).
\end{equation}
Since $\sigma'$ is an arbitrary separable state, we can also take the infimum over this set, obtaining:

\begin{equation}
    \log_2 K\leq \tilde{E}_\alpha(\sigma^{\otimes n}) +   \frac{\alpha}{\alpha -1} \log_2(\frac{1}{1-\epsilon}).
\end{equation}

It suffices to rewrite it as follows:
\begin{equation}
1 - \epsilon \leq 2^{-n(\frac{\alpha -1}{\alpha}) (R - \frac1n \tilde{E}_\alpha (\sigma^{\otimes n}))}. \end{equation}
Since $1-\epsilon = \Tr \tilde{\sigma} \psi_+^{nR}$, the assertion follows.
\end{proof}
Owing to the above lemma, we can formulate a bound on the relative entropy distance from the set of states constructed from an arbitrary state $\rho$ by admixing a body of separable states.
\begin{equation}
    \mbox{Cone}(SEP,\rho):=\{\tilde{\sigma} |\exists_{p \in [0,1]},\, \tilde{\sigma} = p\rho +(1-p) \sigma,\,\sigma \in SEP\}
\end{equation}
We denote this set as $\mbox{Cone}(SEP,\rho)$ and the relative entropy distance from it by $E_R^{\mbox{Cone}(\rho)}$ or $E_R^{\mbox{Cone}}$ when $\rho$ is known from the context.
\begin{lemma}
\label{lem:er}
Let $\sigma_{AA'BB'C}$ be an arbitrary quantum state. For a maximally entangled state $\psi_+^{Rn}$, a state $\tilde{\sigma}_n := \Tr_{A'B'}\tau\Tr_C \Lambda (\sigma^{\otimes n})\tau^{\dagger}$  where $\tau$ is a twisting, and every $\alpha \in (1,\infty)$, there is
\begin{align}
&E_R^{\mbox{Cone}(\tilde{\sigma}_n)}(\psi_+^{Rn})\equiv\inf_{\sigma_n\in \mbox{Cone}(SEP,\tilde{\sigma}_n)}D(\psi_+^{Rn}||\sigma_n) \geq\nonumber\\& 
n(\frac{\alpha -1}{\alpha}) (R - \frac1n \tilde{E}_\alpha (\sigma^{\otimes n})).
\end{align}
\end{lemma}
\begin{proof}
We follow the proof of Lemma $7$ of \cite{keyhuge}. Let $\sigma_n\in \mbox{Cone}(SEP,\tilde{\sigma}_n)$. Using the concavity of the logarithm function, we first
note that 
\begin{align}
    & D(\psi_+^{Rn}||\sigma_n) = \underbrace{\Tr \psi_+^{Rn} \log \psi_+^{Rn}}_{= 0} -\Tr \psi_+^{Rn} \log \sigma_n \geq \nonumber\\
    & \qquad \qquad -\log \Tr \psi_+^{Rn}\sigma_n
\end{align}
Since $\sigma_n \in \mbox{Cone}(SEP,\tilde{\sigma}_n)$, it can be expressed as $\sigma_n = p\tilde{\sigma}_n + (1-p)\sigma_{SEP}$ for some $p\in [0,1]$ and some $\sigma_{SEP}\in SEP$. Thus, we can rewrite the above equation as follows
\begin{align}
    & D(\psi_+^{Rn}||\sigma_n) \geq  -\log \Tr \psi_+^{Rn}\left(p\tilde{\sigma}_n + (1-p)\sigma_{SEP}\right) = \\
    & -\log \left(p\Tr \psi_+^{Rn}\tilde{\sigma}_n + (1-p)\Tr\psi_+^{Rn}\sigma_{SEP}\right) \geq\\
    & \qquad -\log \left(\max\{\Tr \psi_+^{Rn}\tilde{\sigma}_n, \Tr\psi_+^{Rn}\sigma_{SEP}\}\right)
\end{align}
Lemma 7 of \cite{keyhuge} states that for a separable state $\sigma_{SEP}$ there is $\Tr\psi_+^{Rn}\sigma_{SEP} \leq 2^{-Rn}$. Note that since $Rn \geq n(\frac{\alpha-1}{\alpha}) (R - \frac1n \tilde{E}_\alpha (\sigma^{\otimes n}))$, then $\Tr\psi_+^{Rn}\sigma_{SEP} \leq 2^{-n(\frac{\alpha-1}{\alpha}) (R - \frac1n \tilde{E}_\alpha (\sigma^{\otimes n}))}$. Moreover lemma \ref{lem:tr} states that $\Tr \psi_+^{Rn}\tilde{\sigma}_n \leq 2^{-n(\frac{\alpha-1}{\alpha}) (R - \frac1n \tilde{E}_\alpha (\sigma^{\otimes n}))}$. Thus, we can combine all of that to obtain
\begin{align}
    & D(\psi_+^{Rn}||\sigma_n) \geq -\log \left(\max\{\Tr \psi_+^{Rn}\tilde{\sigma}_n, \Tr\psi_+^{Rn}\sigma_{SEP}\}\right) \geq \\
    & -\log 2^{-n(\frac{\alpha-1}{\alpha}) (R - \frac1n \tilde{E}_\alpha (\sigma^{\otimes n}))} = n(\frac{\alpha-1}{\alpha}) (R - \frac1n \tilde{E}_\alpha (\sigma^{\otimes n})),
\end{align}
which is true for every $\sigma_n \in \mbox{Cone}(SEP,\tilde{\sigma}_n)$. Thus, taking the infimum on the left-hand side completes the proof.

\end{proof}

We are ready to state the main theorem. It shows that the one-way repeater rate is upper bounded by the locally measured regularized relative entropy of entanglement (increased by a factor $\frac{\alpha}{\alpha -1 }$) in addition to the sandwich relative entropy of entanglement. In what follows, we will use notation
$\rho\equiv \rho_{AC_1}$ and $\rho'\equiv \rho'_{C_2B}$, as well
as $\rho''\equiv \rho''_{ABC_1C_2}$ and $C_1C_2\equiv {\bf C}$.
\begin{theorem}
For any states $\rho_{AC_1},\rho'_{BC_2},\rho''_{ABC_1C_2}$ and a real parameter $\alpha \in (1,\infty)$, there is
\begin{equation}
    R^{\rightarrow}(\rho,\rho')\leq \frac{\alpha}{\alpha -1}D_{\bf{C}}^{\infty}(\rho\otimes \rho'||\rho'')+ \tilde{E}_\alpha (\rho'').
\end{equation}
\label{thm:main}
\end{theorem}
\begin{proof}
Our proof closely follows \cite{CF17} (Corollary 25), with 
suitable modifications. We first note that there is the following chain of (in)equalities, as it is noted in \cite{CF17}:
\begin{align}
&{\frac1n}D_{\bf{C}}(\rho^{\otimes n}\otimes \rho'^{\otimes n}||\rho''^{\otimes n})\geq \\
&{\frac1n}D({\cal M}_{\bf{C}}(\rho^{\otimes n}\otimes \rho'^{\otimes n})||{\cal M}_{\bf{C}}(\rho''^{\otimes n}))\geq \\
    & {\frac1n}D(\Lambda \circ{\cal M}_{\bf{C}}(\rho^{\otimes n}\otimes \rho'^{\otimes n})||\Lambda \circ{\cal M}_{\bf{C}}(\rho''^{\otimes n}))\geq \label{eq:starting}\\
&\frac1n D(\tilde{\gamma}^{Rn}||\tilde{\sigma}_n),
\end{align}
where $\Lambda$ is an optimal one-way protocol
realizing $R^{\rightarrow}(\rho,\rho')$, and $\tilde{\gamma}^{Rn}$ is a state close to some private state $\gamma^{Rn} = (\tau |\psi_+\>\<\psi_+|\otimes \rho_{A'B'} \tau^{\dagger})^{Rn}$ by $\epsilon$. The state $\tilde{\sigma}_n$ is equal to $\Tr_C\mathcal{L}(\rho''^{\otimes n})$, where $\mathcal{L}$ is an arbitrary operation from $LOCC({\bf C}\rightarrow (A\leftrightarrow B))$.
In the above, we have used the monotonicity of the locally measured, regularized relative entropy of entanglement under the joint application of a CPTP map.

We further have the following chain of (in)equalities:
\begin{align}
&    {\frac1n}D_{\bf{C}}(\rho^{\otimes n}\otimes \rho'^{\otimes n}||\rho''^{\otimes n})\geq  \\
& \frac1n D(\tilde{\gamma}^{Rn}||\tilde{\sigma}_n) \geq\frac1n D({\cal S}(\tilde{\gamma}^{Rn})||{\cal S}(\tilde{\sigma}_n)) \geq \label{eq:psq}\\
& \frac1n \inf_{\sigma\in \mbox{Cone}(SEP,{\cal S}(\tilde{\sigma}_n))}D({\cal S}(\tilde{\gamma}^{Rn})||\sigma) = \frac1n E_R^{\mbox{Cone}}({\cal S}(\tilde{\gamma}^{Rn})) \label{eq:relax}
\end{align}
In equation (\ref{eq:psq}), we use joint monotonicity of the relative entropy under
a privacy squeezing map ${\cal S}(\cdot) = \Tr_{A'B'}\tau^\dagger(\cdot)\tau$, i.e., the inverse operation to twisting $\tau$ composed with tracing out the shielding system. The subsequent inequality follows from the relaxation of entropy distance to the infimum over a convex
set - the cone of states obtained by admixing 
any separable state to the state ${\cal S}(\tilde{\sigma}_n)$. The latter state forms
the apex of this set.

We continue lowerbounding ${\frac1n}D_{\bf{C}}(\rho^{\otimes n}\otimes \rho'^{\otimes n}||\rho''^{\otimes n})$ as follows
\begin{align}
& {\frac1n}D_{\bf{C}}(\rho^{\otimes n}\otimes \rho'^{\otimes n}||\rho''^{\otimes n})\geq \\
& \frac1n E_R^{\mbox{Cone}}({\cal S}(\tilde{\gamma}^{Rn}))\geq \label{eq:relative-entropy-asymptotic-cont}\\
& \frac1n \left[E_R^{\mbox{Cone}}(\psi_+^{Rn}) - O(\epsilon) - h\left(\frac{\varepsilon}{1 + \varepsilon}\right) \right] \geq \\
&\frac1n \left[n\frac{\alpha -1}{\alpha}\left(R - \frac1n \tilde{E}_\alpha(\rho''^{\otimes n})\right)\right] - \\
& \qquad \qquad {1\over n} \left[O(\epsilon) + h\left(\frac{\varepsilon}{1 + \varepsilon}\right) \right] \geq \\
& \frac{\alpha -1}{\alpha}\left(R - \tilde{E}_\alpha(\rho'')\right) - {1\over n} \left[O(\epsilon) + h\left(\frac{\varepsilon}{1 + \varepsilon}\right) \right].
\end{align}
Inequality \eqref{eq:relative-entropy-asymptotic-cont} follows from the fact that the relative entropy distance from a bounded convex set containing a maximally mixed state is asymptotically continuous (see Lemma 7 \cite{Winter_2016}). Next, since ${\cal S}(\tilde{\sigma}_n) = \Tr_{A'B'}\tau^\dagger\Tr_C\mathcal{L}(\rho''^{\otimes n})\tau$, we use Lemma \ref{lem:er}. The last inequality follows from subadditivity of $\tilde{E}_\alpha$ (see Eq. $5.26$ of \cite{WTB2017}). By rearranging the obtained inequality, we get
\begin{align}
    & \frac{\alpha}{\alpha - 1}{\frac1n}D_{\bf{C}}(\rho^{\otimes n}\otimes \rho'^{\otimes n}||\rho''^{\otimes n}) + \tilde{E}_\alpha(\rho'') \geq  \\
    & \qquad \qquad \qquad R - {1\over n} \frac{\alpha}{\alpha - 1} \left[O(\epsilon) + h\left(\frac{\varepsilon}{1 + \varepsilon}\right) \right].
\end{align}
Taking the limit $n\to \infty$, we end the proof.

\end{proof}
From the Theorem \ref{thm:main} and the main result of \cite{CF17}, which states that any key correlated state $\rho$ and its attacked version $\hat \rho$ satisfy 
\begin{equation}
    E_D^{A\rightarrow B}(\rho_{AB}) \geq  D_{A}^{\infty}(\rho||{\hat \rho}),
    \label{eq:CF_bound}
\end{equation}
we have an immediate corollary stated below.
\begin{corollary} 
For a key-correlated states $\rho_{AC_1},\rho'_{BC_2}$ given in Eq. (\ref{eq:key_correlated}) and its attacked version $\hat{\rho}_{AC_1},\hat{\rho}'_{BC_2}$ there is

\begin{equation}
    R^{\rightarrow}(\rho,\rho') \leq \frac{\alpha}{\alpha -1 } E_D^{\rightarrow }(\rho\otimes \rho') + \min\{\tilde{E}_\alpha(\hat{\rho}), \tilde{E}_\alpha(\hat{\rho}') \}.
    \label{eq:simplified1}
\end{equation}

\label{cor:simplified1}
\end{corollary}
\begin{proof}
We identify $A$ and $B$ from Eq. (\ref{eq:CF_bound}) with ${\bf C}\equiv C_1C_2$, and $AB$ respectively, to bound $D^{\infty}_{\bf C}(\rho\otimes \rho' \| \hat{\rho}\otimes \hat{\rho}')$ by $E_D^{\rightarrow}(\rho\otimes \rho')$ from above. We further use Theorem
\ref{thm:main} to obtain the following inequality
\begin{equation}
    R^{\rightarrow}(\rho,\rho') \leq \frac{\alpha}{\alpha -1 } E_D^{\rightarrow }(\rho\otimes \rho') + \tilde{E}_\alpha(\hat{\rho}\otimes \hat{\rho}').
\end{equation}
Observe now that in all of the previous proofs we define $\tilde{E}_\alpha$ as an infimum with respect to states separable in $AC_1C_2:B$ cut. Thus we can bound $\tilde{E}_\alpha(\hat{\rho}_{AC_1}\otimes \hat{\rho}'_{BC_2})$ as follows 
\begin{align}
    &\tilde{E}_\alpha(\hat{\rho}_{AC_1}\otimes \hat{\rho}'_{BC_2})\leq \\
    &\qquad \tilde{D}_\alpha(\hat{\rho}_{AC_1}\otimes \hat{\rho}'_{BC_2} \| \hat{\rho}_{AC_1}\otimes \sigma_{BC_2}) \\
    &\qquad \qquad \underbrace{=}_{\text{data processing \cite{WWY14,MDSFT}}} \tilde{D}_\alpha(\hat{\rho}'_{BC_2} \| \sigma_{BC_2}),
\end{align}
where $\sigma_{BC_2}$ is arbitrary separable state. Since it holds true for any choice of separable state $\sigma_{BC_2}$, we can write
\begin{align}
    & \tilde{E}_\alpha(\hat{\rho}_{AC_1}\otimes \hat{\rho}'_{BC_2})\leq \\
    & \qquad \inf_{\sigma_{BC_2}\in\mathrm{SEP}} \tilde{D}_\alpha(\hat{\rho}'_{BC_2} \| \sigma_{BC_2}) = \tilde{E}_\alpha(\hat{\rho}'_{BC_2}),
\end{align}
which implies that 
\begin{equation}\label{eq:bound-on-repeater-rate-b-cut}
    R^{\rightarrow}(\rho,\rho') \leq \frac{\alpha}{\alpha -1 } E_D^{\rightarrow }(\rho\otimes \rho') + \tilde{E}_\alpha(\hat{\rho}').
\end{equation}
On the other hand, notice that our considerations remain true if we define $\tilde{E}_\alpha$ as an infimum with respect to states separable in $A:C_1C_2B$ cut. Thus, by similar arguments, we can show that
\begin{equation}\label{eq:bound-on-repeater-rate-a-cut}
    R^{\rightarrow}(\rho,\rho') \leq \frac{\alpha}{\alpha -1 } E_D^{\rightarrow }(\rho\otimes \rho') + \tilde{E}_\alpha(\hat{\rho}).
\end{equation}
By combining equations \eqref{eq:bound-on-repeater-rate-b-cut} and \eqref{eq:bound-on-repeater-rate-a-cut} we get
\begin{equation}
    R^{\rightarrow}(\rho,\rho') \leq \frac{\alpha}{\alpha -1 } E_D^{\rightarrow }(\rho\otimes \rho') + \min\{\tilde{E}_\alpha(\hat{\rho}), \tilde{E}_\alpha(\hat{\rho}') \},
\end{equation}
which finishes the proof.

\end{proof}

We can introduce some simplifications by considering the extremal value of $\alpha$ and the behavior of the entanglement monotones involved in the above corollary. This leads us to a relaxed bound for a single state.

\begin{corollary} For a key correlated state $\rho$ given in Eq. (\ref{eq:key_correlated}) and its key attacked version $\hat \rho$ there is
\begin{equation}
    R^{\rightarrow}(\rho,\rho) \leq 2E_D^{\rightarrow }(\rho) + E_{max}(\hat{\rho}).
\label{eq:simplify2}
\end{equation}

\label{cor:simplify2}
\end{corollary}
\begin{proof}
It follows straightforwardly from the Corollary \ref{cor:simplified1}. Indeed, consider first the case $\rho=\rho'$, and note that $\tilde{E}_\alpha$ is subadditive (see Eq. $5.26$ of \cite{WTB2017}) and $E_D^{\rightarrow}$ is extensive by definition. This implies 
\begin{equation}
    R^{\rightarrow}(\rho,\rho) \leq 2({\alpha \over \alpha -1})E_D^{\rightarrow }(\rho) + \tilde{E}_\alpha(\hat{\rho}).
    \label{eq:single_state_bound}
\end{equation}

Relaxing the above inequality by taking the limit of $\alpha \rightarrow \infty$, we arrive at the claim.
\end{proof}

A special case of a key correlated state is a private bit (a private state with a key part of local dimension $d_k=2$).
When the private state is taken at random, we can then upper bound the term $E_{max}$ to obtain the bound
which is dependent only on the one-way distillable entanglement and a constant factor. We thus arrive at the second main result of this section.
\begin{theorem}
For a random private bit (a private state with $d_k=2$ and arbitrary finite dimensional shield part $d\otimes d<\infty$) there is
\begin{equation}
    R^{\rightarrow}(\rho,\rho) \leq 2E_D^{\rightarrow}(\rho) + 1.
    \label{eq:first_bound}
\end{equation}

\label{thm:first_bound}
\end{theorem}
\begin{proof}
In what follows, for the ease of notation, we will refer to $d_s$, i.e., the dimension of the shielding system of the key correlated state, as $d$.
We will relax the upper bound given in Eq. (\ref{eq:simplify2}) provided in the Corollary \ref{cor:simplify2}.
Let us recall first that $E_{max}(\rho):= \inf_{\sigma \in SEP} \inf\{\lambda \in{\mathrm R}: \rho \leq 2^{\lambda}\sigma\}$.
As a further upper bound, one can use the norm $||.||_{\infty}$. This is because
$E_{max}(\rho) \leq \log_2 \left(d^2||\rho||_\infty\right)$. The latter inequality follows from the fact that the maximally mixed state $\id \over d^2$ is separable and majorizes every state. We then set $\sigma=\frac{\id}{d^2}$ in the definition of $E_{max}$. If we have then an upper bound on the maximal eigenvalue of $\rho$, denoted as $\kappa_{max}$  the value $\lambda = \log_2(d^2\kappa_{max}))$ satisfies $\rho \leq 2^{\lambda} \id/d^2$, and hence $E_{max} \leq \log_2(d^2\kappa_{max})$. It is known~\cite{marvcenko1967distribution}, that the maximal eigenvalue of a random state $\rho_{rand}$ of dimension $d^2$ is upper bounded by 
\begin{equation}
    \lambda_{max}(\rho_{rand}) \leq {4\over d^2}.
\end{equation}
The key attacked state $\hat{\rho}$ of $\rho$ is of the form
\begin{equation}
   \hat{\rho} = {1\over 2}|00\>\<00|\otimes \rho_1 + 
    {1\over 2}|11\>\<11|\otimes \rho_2, 
\end{equation}
where $\rho_i$ are also random states.
Hence, its maximal eigenvalue is upper bounded by ${1\over 2d^2}$. We have then that $E_{max} \leq \log_2 (d^2 \times ({2\over d^2}))=1$, hence the assertion follows.
\end{proof}

Numerical intuition for justification of setting $\sigma=\frac{\id}{d^2}$ in the proof of Theorem~\ref{thm:first_bound}
to upper bound $E_{max}$ for a random state is presented in Fig.~\ref{fig:numerical-intuition}.
The top plot shows the smallest eigenvalue of $2^\lambda \sigma - \rho$ for a fixed random $\rho$ of order $d^2$ and $100$ randomly
chosen $\sigma$ for dimensions $d=8, 16, 32$. As $d$ increases, the eigenvalues become positive for $\lambda \geq 2$. To drive the point on this intuition
further, the bottom plot in Fig.~\ref{fig:numerical-intuition} shows a logarithmic plot
of the absolute value of the maximum of these eigenvalues over 100 randomly chosen $\sigma$ and $\lambda=2$
as a function of dimension $d$.

\begin{figure}[!htp]
    \centering
    \includegraphics[width=0.48\textwidth]{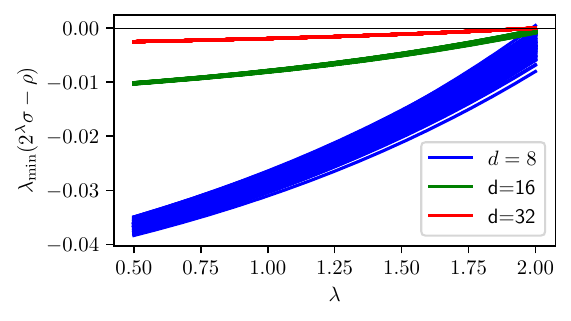}
    \includegraphics[width=0.48\textwidth]{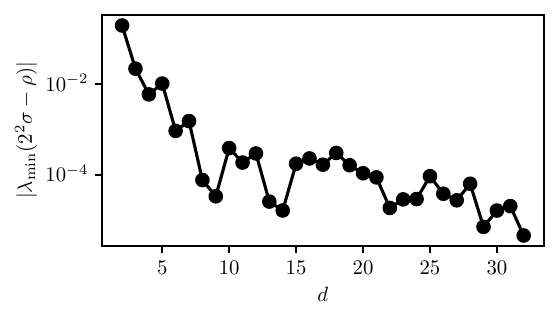}
    \caption{Minimal eigenvalue of matrix $2^\lambda \sigma - \rho$ as a function of $\lambda$ plotted
    for a fixed random $\rho$ of order $d^2$ and 100 randomly selected separable $\sigma$ (top) and  a logarithmic plot
    of the absolute value of the maximum of these eigenvalues over 100 randomly chosen $\sigma$ and $\lambda=2$
    as a function of dimension $d$ (bottom).
    As can be seen, the plots seem to justify the intuition behind the assumption $\lambda=2$ in
    the proof of Theorem~\ref{thm:first_bound}.}
    \label{fig:numerical-intuition}
\end{figure}

The bound given in Eq. (\ref{eq:first_bound}) is loose as it has a non-vanishing factor of $1$. However, it is valid for any dimension of the shielding system of the state, and this factor is independent of the dimension. It implies that no matter how large is the shielding system of a key correlated state, generically, it brings in at most $1$ bits of the repeated key above the (twice) one-way distillable entanglement.

Besides that, there is an independent
bound (see Corollary \ref{cor:tighter}), which can be obtained not from the upper bounding via relative entropy but squashed entanglement, \cite{Christandl2017,Tucci}. In the next section, we will show that the (half of) the {\it mutual information} of a random private state does not exceed the minimal value, which is key by more than $\frac{1}{4\ln 2}$. The bound will then follow from a trivial inequality $R^\rightarrow(\rho,\rho) \le R(\rho,\rho)\leq K_D(\rho)$ \cite{BCHW15}.

\section{Mutual information bound for the secure content of a random private bit}
\label{sec:key_bound}
This section focuses on random private bits and provides an upper bound on the distillable key.
We perform the random choice of a private state differently from that in \cite{Christandl2020}. There, the private state
has been first transformed via a one-way LOCC protocol
to the form of the so-called Bell-private bit of the form $\frac12\left(|\psi_+\>\<\psi_+|\otimes \rho_+ + |\psi_-\>\<\psi_-|\otimes \rho_-\right)$. Then the two states were defined as $\rho_{\pm}:= U\bar{\rho}_{\pm}U$, where $\bar{\rho}_{\pm}$ is arbitrarily chosen state of the form $\bar{\rho}_+:=\frac{Q}{Tr Q}$ with $Q$ being a projector onto any $d^2/2$ dimensional subspace of $A'B'$ system, and $\bar{\rho}_-:=\frac{Q^{\perp}}{Tr Q^{\perp}}$ with a projector $Q^\perp$ complementary to $Q$. Our approach is different, as 
it also applies to the smallest dimension of the shield $d_s=2$. We do not transform a private bit
into a Bell private bit. We rather observe that any private bit is uniquely represented by a (not necessarily normal) operator $X$ of trace norm equal to 1. We then represent this operator
as $U\sigma$ and choose a random state $\sigma$ and random unitary $U$ according to the Haar measure.

We begin by introducing the notion of {\it random} private states. Looking at Eq.~\eqref{eq:private-state}, we have two objects that can be chosen at random:
the unitaries $U_0, U_1$ and the shared state $\sigma_{A'B'}$. Due to the unitary invariance of the ensemble of random quantum states, we can set one of the unitaries
$U_0 = \id$ and the second one to be chosen randomly, denoted $U_1=U$. For simplicity, let us also introduce the notation $X= U \sigma_{A'B'}$. Hence, we can rewrite
Eq.~\eqref{eq:private-state} as
\begin{equation}
\gamma_{ABA'B'} =  {1\over 2} 
\begin{pmatrix}
\sqrt{XX^\dagger} & 0 & 0 & X\\
0 & 0 & 0 & 0 \\
0 & 0 & 0 & 0 \\
X^\dagger & 0 & 0 & \sqrt{X^\dagger X}
\end{pmatrix}.\label{eq:gamma}
\end{equation}
A short introduction to the properties of random unitaries and quantum states can be found in Appendix~\ref{app:random}.
Before the statement of the main result, which is the bound on the distillable key of a random private bit, let us first state the following technical fact
\begin{proposition}\label{th:eigvals}
    Let $\gamma_{ABA'B'}$ be defined as in Eq.~\eqref{eq:gamma}, let $X$ be an arbitrary
    random matrix of dimension $d^2$ and let $\mu_X$ denote the asymptotic
    eigenvalue density of $\sqrt{XX^\dagger}$. Then, as $d\to\infty$, $\gamma$
    has the eigenvalue density $\mu_\gamma$ given by
    \begin{equation}
        \mu_\gamma = \frac12 \mu_X + \frac12 \delta(0).
    \end{equation}
\end{proposition}
\begin{proof}
    From the form of $\gamma_{ABA'B'}$ we can conclude that half of its
    eigenvalues are equal to zero. Therefore, it suffices to focus on the
    eigenvalues of the matrix
    \begin{equation}\label{eq:simple-proof}
        \gamma' = \frac12 \begin{pmatrix}
            \sqrt{XX^\dagger} & X \\
            X^\dagger & \sqrt{X^\dagger X}.
        \end{pmatrix}.
    \end{equation}
First we note that for any $\ket{\psi}$ we have
\begin{equation}
    \gamma' \begin{pmatrix}
    \ket{\psi} \\
    -X^\dagger (\sqrt{XX^\dagger})^{-1} \ket{\psi}
    \end{pmatrix} = 0.
\end{equation}
Hence $\rank{\gamma'}=d^2$, which gives us that half of the eigenvalues of $\gamma'$ are equal to zero.
From the fact that the Schur complement of the first block is $\gamma'/\sqrt{XX^\dagger}=0$, we recover
that the remaining eigenvalues of $\gamma'$ are those of $\sqrt{XX^\dagger}$.
\end{proof}

Additionally, we will need the following fact regarding the entropy of random quantum states sampled
from the Hilbert-Schmidt distribution from~\cite{bengtsson2017geometry,collins2011gaussianization,puchala2016distinguishability,zyczkowski2001induced}. 

\begin{proposition}\label{th:entropy}
    Let $\rho$ be a random mixed state of dimension $d$ sampled from the Hilbert-Schmidt distribution.
    Then, for large $d$ we have
    \begin{equation}
        \mathbb{E}\left( S(\rho) \right) = \log_2 d - \frac{1}{2 \ln 2} - O\left( \frac{\log_2 d}{d} \right)
    \end{equation}
\end{proposition}
\begin{proof}
    For large $d$ it can be shown~\cite{zyczkowski2001induced} that
    \begin{equation}
        \mathbb{E}(\Tr \rho^k) = d^{1-k} \frac{\Gamma(1+2k)}{\Gamma(1+k) \Gamma(2+k)} \left( 1 + O \left( \frac{1}{d} \right) \right).
    \end{equation}
    What remains is to note that
    \begin{equation}
        \mathbb{E}\left( S(\rho) \right) = - \lim_{k \to 1} \frac{\partial \Tr \rho^k}{\partial k},
    \end{equation}
    and the desired results follow from direct calculations.
\end{proof}

The main result of this section can be formulated as the following theorem
\begin{theorem}
    Let $\gamma_{ABA'B'}$ be a private state defined as in Eq.~\eqref{eq:gamma}, with the shielding system of dimension $d_s\otimes d_s$, let also $X=U\sigma_{A'B'}$,
    where $U$ is a Haar unitary be an arbitrary and $\sigma_{A'B'}$ is a random mixed
    state sampled from the Hilbert-Schmidt distribution. Then as $d\to\infty$ we get
    \begin{equation}
        I(AA':BB')_{\gamma} \to 2 + \frac{1}{2\ln 2}.
    \end{equation}
\end{theorem}
\begin{proof}
We first justify the entropy of the full state $\gamma$. Recall that $X=U\sigma_{A'B'}$ with $\sigma_{A'B'}\ge 0$ and $\Tr\sigma_{A'B'}=1$.
Then
\begin{equation}
\sqrt{XX^\dagger}=\sqrt{U\sigma_{A'B'}^2U^\dagger}=U\sigma_{A'B'}U^\dagger,
\end{equation}
hence $\sqrt{XX^\dagger}$ and $\sigma_{A'B'}$ have identical spectra. By Proposition~\ref{th:eigvals}, the non-zero eigenvalues of $\gamma$ coincide with those of $\sqrt{XX^\dagger}$ (the remaining eigenvalues are zeros, which do not contribute to the von-Neumann entropy). Therefore
\begin{equation}
S(AA'BB')_\gamma = S(\gamma) = S(\sqrt{XX^\dagger}) = S(\sigma_{A'B'}).
\end{equation}
Applying Proposition~\ref{th:entropy} with $d=d_s^2$ yields
\begin{equation}\label{eq:global-entropy-random-pdit}
    S(AA'BB')_{\gamma} = \log_2 d_s^2 - \frac{1}{2 \ln 2} - O\left(\frac{\log_2 d_s^2}{d_s^2}\right).
\end{equation}

All is left is to consider the terms  $S(AA')_{\gamma}$ and $S(BB')_{\gamma}$. First, we observe that
\begin{equation}
    \gamma_{AA'} = \frac12 \begin{pmatrix}
        \Tr_{B'} U \sigma_{A'B'} U^\dagger & 0\\
        0 & \Tr_{B'}\sigma_{A'B'}
    \end{pmatrix},
    \label{eq:rand_subsystem}
\end{equation}
and similarly for $\gamma_{BB'}$. As the ensemble of random quantum states is
unitarily invariant, $U\sigma_{A'B'}U^\dagger$ is also a random quantum state.
Based on the results from~\cite{nechita2018almost} we have that the partial
trace of a $d_s^2 \times d_s^2$ random quantum state is almost surely the maximally
mixed state. Hence
\begin{equation}\label{eq:entropy-of-a-part-pdit}
    S(AA')_{\gamma} = S(BB')_{\gamma} = \log_2 2d_s - o(1).
\end{equation}
Putting all of these facts together, we obtain the desired result
\begin{equation}
    I(AA':BB')_{\gamma} \to 2\log_2 2 + \frac{1}{2 \ln 2}.
    \label{eq:bound}
\end{equation}
\end{proof}
We then have the immediate corollary.
\begin{corollary}
\label{cor:tighter}
For a random private bit $\gamma_2$ of an arbitrarily large dimension
of the shield, in the limit of large $d_s \rightarrow \infty$, there is
\begin{equation}
    K_D(\gamma_2) \leq \log_2 2 +\frac{1}{4 \ln 2}
\end{equation}
\end{corollary}
\begin{proof}
We use the bound by squashed entanglement $E_{sq}(\rho):= \inf_{\rho_{ABE}, Tr_E\rho_{ABE}=\rho_{AB}}\frac{1}{2}I(A:B|E)_{\rho_{ABE}}$ \cite{squashed}, by noticing,
that it is upper bounded by $\frac{1}{2}I(A:B)$ (we consider half of the
bound on the mutual information given in (\ref{eq:bound})).
\end{proof}

For the private states, the above bound yields an independent result than the one given in Theorem \ref{thm:first_bound}.
Indeed, the one-way repeater rate of a bipartite state $\rho_{AB}$
can not be larger than the distillable key $K_D(\rho_{AB})$ \cite{BCHW15}. This is a trivial bound: any key-repeater protocol can be viewed as a particular LOCC protocol between $A$ and $B$. As such, it can
not increase the initial amount of key in the cut $A:B$. We thus have for a generic  private state $\gamma$ (not necessarily irreducible) a trivial bound
\begin{align}
    &R^{\rightarrow}(\gamma_{2},\gamma_{2})\leq K_D(\gamma_{2})\leq \nonumber\\&1 + {1\over 4 \ln 2}\approx 1 +0.360674.
\end{align}
Let us note here that the bound
is far from being small, in contrast to the bound for the private state taken at random, due to the different randomization procedures proposed in \cite{Christandl2020}. This is because our technique does not immediately imply that the private state has Hermitian $X$,
nor that its positive and negative parts are random states.
We believe, however, that the above bound can be made tighter.


\section{Localizable private randomness for a generic local independent state}
\label{sec:ibits}
In this section, we study the rates of private randomness that can be distilled from a generic independent bit. We base our results on the following idea. 
While (half of the) mutual information is merely a weak bound on the private key, this function reports the exact amount of private randomness content of a quantum state in various scenarios \cite{YHW19}.

Let us recall here that the local independent bit (in a bipartite setting) has the form
\begin{equation}
    \alpha_{AA'B'}= \sum_{i,j=0}^{1}{1\over 2}|i\>\<j|_A\otimes U_i \sigma_{A'B'} U_j^{\dagger}.
    \label{eq:random_alpha}
\end{equation}
where $U_i$ are unitary transformations acting on system $A'B'$ and $\sigma$ is an arbitrary state on the latter system. We assume that the system $A'B'$ is of dimension $d_s\otimes d_s$. 
In the above we borrow the notation of \cite{YHW19} where $\alpha$ denotes the independent states, thus $\alpha$ should not be confused with the parameter of the sandwiched R\'enyi relative entropy 
given in Eq. (\ref{eq:Renyi}).
In the matrix form, it is as follows
\begin{equation}\label{eq:ibit}
   \alpha_{AA'B'}={1\over 2} \left(\begin{array}{cc}
    \sqrt{XX^{\dagger}}    & X \\
        X^{\dagger} & \sqrt{X^{\dagger}X}
    \end{array}\right)
\end{equation}
where $X = U_0\sigma_{A'B'}U_1$. 
As it was observed in the case of the private bit, we can safely assume that $U_0 = \id$,
and $U_1$ is arbitrary because $\sigma_{A'B'}$ will be taken at random.

\begin{proposition}
\label{prop:last}
    Consider $\alpha_{AA'B'}$ as in Eq.~(\ref{eq:ibit}). Then, $\tr_{B'} \alpha_{AA'B'}$ converges almost surely to the maximally mixed state as $d \to \infty$
    \begin{equation}
        \lim_{d\to\infty} \|2d \tr_{B'}\alpha_{AA'B'} - \1 \| =0.
    \end{equation}
    Moreover, we have
    \begin{equation}
        \|2d \tr_{B'}\alpha_{AA'B'} - \1 \|_\infty = O(d^{-1/2})
    \end{equation}
    \begin{proof}
        We start by proving the limit. The explicit form of the partial trace reads
        \begin{equation}
            \tr_{B'} \alpha_{AA'B'} = \frac12 \begin{pmatrix}
                    \tr_{B'} \sigma_{A'B'} & \tr_{B'} U_0 \sigma_{A'B'} U_1 \\
                    \tr_{B'} U_1^\dagger \sigma_{A'B'} U_0^\dagger & \tr_{B'} U_1^\dagger \sigma_{A'B'} U_1
            \end{pmatrix}.
        \end{equation}
        We immediately note that the diagonal blocks are partial traces of independent random states. Hence, by~\cite{nechita2018almost}, they almost surely converge to the maximally mixed state. Formally, we have (only for the first block)
        \begin{equation}
            \lim_{d\to\infty}\| d \tr_{B'} \sigma_{A'B'} - \1\| = 0
        \end{equation}
            It remains to control the off-diagonal block $B:=\tr_{B'}(U_0\sigma_{A'B'}U_1)$ (we set $U_0=\id$).
            Conditioned on $\sigma_{A'B'}$, Haar invariance implies $\mathbb{E}_{U_1}[B]=0$. Moreover, using the second moment identity
            $\mathbb{E}_{U_1}[U_{ab}\overline{U_{cd}}]=\delta_{ac}\delta_{bd}/d^2$ for $U_1\in U(d^2)$ and a direct index calculation, we obtain
            \begin{equation}
            \mathbb{E}_{U_1}\|B\|_2^2=\frac{1}{d}\Tr(\sigma_{A'B'}^2).
            \end{equation}
            For $\sigma_{A'B'}$ sampled from the Hilbert--Schmidt distribution on $\mathbb{C}^{d^2}$ we have $\mathbb{E}\Tr(\sigma_{A'B'}^2)=O(d^{-2})$ (see, e.g., induced-state moment formulas in~\cite{zyczkowski2001induced}), hence
            $\mathbb{E}\|B\|_2^2=O(d^{-3})$ and therefore $\|B\|_\infty\le \|B\|_2 = O(d^{-3/2})$.
            Consequently $d\|B\|_\infty = O(d^{-1/2})\to 0$.

            The diagonal blocks are reduced states of independent Hilbert--Schmidt random states; by concentration for induced states (see~\cite{nechita2018almost}),
            \begin{equation}
            \begin{split}
            \|d\,\tr_{B'}\sigma_{A'B'}-\1\|_\infty & = o(1),\\
            \|d\,\tr_{B'}(U_1^\dagger\sigma_{A'B'}U_1)-\1\|_\infty & =o(1),
            \end{split}
            \end{equation}
            almost surely. Putting diagonal and off-diagonal bounds together yields
            \begin{equation}
            \lim_{d\to\infty} \|2d\,\tr_{B'}\alpha_{AA'B'}-\1\|_\infty=0.
            \end{equation}
    \end{proof}
\end{proposition}

In what follows, we will show achievable rates of private randomness for a randomly chosen ibit. They are encapsulated in the following theorem.\\
\begin{theorem}
    Let $\alpha_{AA'B'}$ given by Eq. (\ref{eq:random_alpha}) be the randomly generated state by picking up a random state $\sigma_{A'B'}$ and a unitary transformation $U_1$ due to Haar measure. Let also $R_G := 1 + \frac{1}{2\ln 2}$.
    Then for all but finitely many $d_s\in \mathbb{N}$ the following bounds are achievable
    \begin{enumerate}
        \item for no communication and no noise,\\ 
        $R_A < R_G$, $R_B < \frac{1}{2\ln 2}$,\\
        $R_A + R_B \leq R_G$,
        \item for free noise, but no communication, \\
        $R_A < R_G$, $R_B < \frac{1}{2\ln 2}$,\\
        $R_A + R_B \leq R_G$,
        \item for free noise and free communication,\\
        $R_A,R_B \leq  R_G$, $R_A + R_B \leq R_G$,
        \item for free communication but no noise,\\
        $R_A,R_B \leq  R_G$, $R_A + R_B \leq R_G$. 
    \end{enumerate}
\end{theorem}
\begin{proof}
    At first observe that for the considered state $\alpha_{AA'B'}$, we have
    \begin{align}
        &\log_2 |AA'B'| - S(AA'B') = \\
        & \log_2 2d_s^2 - \log_2 d_s^2 + \frac{1}{2\ln 2} + O\left( \frac{\log_2 d_s^2}{d_s^2} \right).
    \end{align}
    Thus, its global purity satisfies
    \begin{equation}
        \log_2 |AA'B'| - S(AA'B') \geq 1 + \frac{1}{2\ln 2}.
    \end{equation}
    In the above we have used Equation \eqref{eq:global-entropy-random-pdit} and the fact, that
    $S(AA'B')$ for an independent state $\alpha_{AA'B'}$ equals
    $S(AA'BB')$ for a corresponding private bit (i.e., generated by the same twisting $\tau$ and from the same state $\sigma$). Hence, $R_G = 1 + \frac{1}{2\ln 2}$ lowerbounds global purity of the considered $\alpha_{AA'B'}$ state.
    We are now ready to prove each case of the theorem separately. \\
    1. By Theorem \ref{thm:priv_rand}, case $1$, the necessary condition for the rate of private randomness $R_A$ to be achievable, is $R_A \leq \log_2 |AA'| - \max\{0, S(AA'|B')\}$. We want to simplify this formula. Observe that 
    \begin{align}
        & S(AA'|B') = S(AA'B') - S(B') \geq \\
        & \log_2 d_s - \frac{1}{2\ln 2} - C_{AA'B'}\frac{\log_2d_s^2}{d_s^2},
    \end{align}
    for some constant $C_{AA'B'}$. In the above, we have used equations \eqref{eq:global-entropy-random-pdit} and \eqref{eq:entropy-of-a-part-pdit} to bound terms $S(AA'B')$ and $S(B')$ respectively. Surely, $\log_2 d_s - \frac{1}{2\ln 2} - C_{AA'B'}\frac{\log_2d_s^2}{d_s^2} \geq 0$ for all but finitely many $d_s$. Thus we can consider only the case in which $\max\{0, S(AA'|B')\} = S(AA'|B')$. Now, we aim to upperbound the $S(AA'|B')$ term.
    Notice that 
    \begin{align}
        & S(AA'|B') = S(AA'B') - S(B') \leq \\
        & \log_2 d_s^2 - \frac{1}{2\ln 2} - \log_2 d_s + f_{B'}(d_s) = \\
        & \log_2d_s - \frac{1}{2\ln 2} + f_{B'}(d_s),
    \end{align}
    for $f_{B'}(d_s) = o(1)$. Hence, we have 
    \begin{align}
        &\log_2 |AA'| - \max\{0, S(AA'|B')\} \geq \\
        &\log_2 2d_s - \log_2d_s + \frac{1}{2\ln 2} - f_{B'}(d_s) = \\
        & 1 + \frac{1}{2\ln 2} - f_{B'}(d_s),
    \end{align}
    so the rate $R_A$ satisfying $R_A \leq 1 + \frac{1}{2\ln 2} - f_{B'}(d_s)$ is achievable. This implies that all rates satisfying $R_A < 1 + \frac{1}{2\ln 2}$ are achievable for all but finitely many $d_s$.
    On the other hand, by analogous arguments, one can show that all rates $R_B$ satisfying 
    \begin{equation}
        R_B \leq \frac{1}{2 \ln 2} - f_{AA'}(d_s),
    \end{equation}
    where $f_{AA'}(d_s) = o(1)$, are achievable. Once again, this implies that all rates $R_B$ satisfying $R_B < 1 + \frac{1}{2\ln 2}$ are achievable for all but finitely many $d_s$. Finally, notice that the pair $(R_A, R_B)$ is achievable only if the constraint $R_A + R_B \leq R_G$ is satisfied.\\
    2. Proof of this case is the same as for case 1.\\
    3. Follows directly from Theorem \ref{thm:priv_rand} case 3.\\
    4. For all but finitely many $d_s$, inequalities $S(AA'B') \geq S(AA')$ and $S(AA'B') \geq S(B')$ are satisfied. Hence, formulas from Theorem \ref{thm:priv_rand} case 4 for achievable rates $R_A$ and $R_B$, reduce to $\log_2 |AA'B'| - S(AA'B')$. As we showed previously, $\log_2 |AA'B'| - S(AA'B') \geq 1 + \frac{1}{2 \ln 2} = R_G$. Thus, for all but finitely many $d_s$ the rates  satisfying $R_A, R_B \leq R_G$, $R_A + R_B \leq R_G$ are achievable.
\end{proof}
From the above theorem, we can see that
in case of the system $A$, for asymptotically large dimension $d_s$ the value 
$1+\frac{1}{2\ln 2}$ 
can be approached in all four cases. In the case of system $B$ we invoked only the (lower and upper) bounds on the entropy of involved systems ($AA'$). Hence, the bounds are not tight, possibly less than the maximal achievable ones.

\section{Discussion}
\label{sec:discussion}
We have generalized the bound by M. Christandl and R. Ferrara \cite{CF17} to the case of arbitrary key correlated states. 
We then show a sequence of relaxation of this bound, from which it follows that the repeated key of a key correlated state can not be larger than twice one-way distillable entanglement plus
the max-relative entropy of its attacked state.

We further ask how big the key content of a random private bit is, which need not be irreducible \cite{HCRS18}. 
A not irreducible private bit can have more distillable key than $1$.
It turned out that the amount of key is bounded by a constant factor independent of the dimension of the shielding system.
It is interesting if the constant $\approx 0.36$ can be improved. One could also 
ask if this randomization technique also results in a state for which the {\it repeated} key is vanishing with a large dimension of the shield, as it was shown by a different technique of randomization in \cite{Christandl2020}. 
Recent results on private randomness generation \cite{YHW19} also allow us to estimate the private randomness content of generic independent bits. Generalizing this result to independent states of larger dimensions would be the next important step. 
Finally, we emphasize that our approach is generic. That is, it appears that any other strong-converse bound on the quantum distillable key may give rise to a new, possibly tighter
bound on the one-way quantum key repeater rate. Hence, recent strong-converse bound on quantum privacy amplification \cite{StrongConverse} paves the way for future research in this direction.
It would be important to extend this technique to the two-way quantum key repeater rate.

\acknowledgements{
KH acknowledges Roberto Ferrara, Siddhartha Das and Marek Winczewski for helpful discussion. KH acknowledges the Fulbright
Programm and Mark Wilde for hospitality during the Fulbright
scholarship at the School of Electric and Computer
Engineering of the Cornell University.
We acknowledge Sonata Bis 5 grant
(grant number: 2015/18/E/ST2/00327) from the National Science Center.
We acknowledge partial support by the Foundation for Polish Science (IRAP project, ICTQT, contract no. MAB/2018/5, co-financed by EU within Smart Growth Operational Programme). The ’International Centre for Theory of Quantum Technologies’ project (contract no. MAB/2018/5) is carried out within the International Research Agendas Programme of the Foundation for Polish Science co-financed by the European Union from the funds of the Smart Growth Operational Programme, axis IV: Increasing the research potential (Measure 4.3). KH acknowledges National Science Centre,
Poland grant OPUS UMO-2023/49/B/ST2/02468. 
}
\bibliographystyle{apsrev4-1}
\bibliography{references.bib}

\appendix

\section{Random quantum objects}\label{app:random}
In this section, we provide a short introduction to random matrices. The scope
is limited to concepts necessary for the understanding of our result.

\subsection{Ginibre matrices}

We start off by introducing the Ginibre random matrices
ensemble~\cite{ginibre1965statistical}. This ensemble is at the core of a vast
majority of algorithms for generating random matrices presented in later
subsections. Let $(G_{ij})_{1 \leq i \leq m, 1 \leq j \leq n}$ be a $m\times n$ table of independent identically distributed (i.i.d.) random variables on
$\mathbb{C}$. The field $\mathbb{K}$ can be either of $\mathbb{R}$, $\C$ or
$\mathbb{Q}$. With each of the fields, we associate a Dyson index $\beta$ equal
to $1$, $2$, or $4$, respectively. Let $G_{ij}$ be i.i.d random variables with
the real and imaginary parts sampled independently from the distribution
$\mathcal{N}(0, \frac{1}{\beta})$. Hence, $G \in \Lrm(\XX, \YY)$, where matrix
$G$ is
\begin{equation}
P(G) \propto \exp(-\Tr G G^\dagger).
\end{equation}
This law is unitarily invariant, meaning that for any unitary matrices $U$ and
$V$, $G$, and $UGV$ are equally distributed. It can be shown that for $\beta=2$ the eigenvalues of $G$ are uniformly distributed over the unit disk on the
complex plane~\cite{tao2008random}.

\subsection{Wishart matrices}

Wishart matrices form an ensemble of random positive semidefinite matrices. They
are parametrized by two factors. First is the Dyson index $\beta$, which is equal
to one for real matrices, two for complex matrices, and four for symplectic
matrices. The second parameter, $K$, is responsible for the rank of the
matrices. They are sampled as follows
\begin{enumerate}
\item Choose $\beta$ and $K$.

\item Sample a Ginibre matrix $G\in \Lrm(\XX, \YY)$ with the Dyson index $\beta$
and $\mathrm{dim}(\XX) = d$ and $\mathrm{dim}(\YY)=Kd$.

\item Return $W = GG^\dagger$.
\end{enumerate}

Sampling this ensemble of matrices will allow us to sample random quantum
states. This process will be discussed in further sections. Aside from their
construction, we will not provide any further details on Wishart matrices, as
this falls outside the scope of this work.

\subsection{Circular unitary ensemble}
Circular ensembles are measures on the space of unitary matrices. Here, we focus
on the circular unitary ensemble (CUE), which gives us the Haar measure on the
unitary group. In the remainder of this section, we will introduce the algorithm
for sampling such matrices.

There are several possible approaches to generating random unitary matrices
according to the Haar measure. One way is to consider known parametrizations of
unitary matrices, such as the Euler~\cite{zyczkowski1994random} or
Jarlskog~\cite{jarlskog2005recursive} ones. Sampling these parameters from
appropriate distributions yields a Haar random unitary. The downside is the long
computation time, especially for large matrices, as this involves a lot of
matrix multiplications. We will not go into this further; instead, we refer the
interested reader to the papers on these parametrizations.

Another approach is to consider a Ginibre matrix $G \in \Lrm(\XX)$ and its polar
decomposition $G=U P$, where $U \in \Lrm(\XX)$ is unitary, and $P$ is a positive
matrix. The matrix $P$ is unique and given by $\sqrt{G^\dagger G}$. Hence, assuming $P$ is invertible, we could recover $U$ as
\begin{equation}
U = G (G^\dagger G) ^{-\frac{1}{2}}.
\end{equation}
As this involves the inverse square root of a matrix, this approach can be
potentially numerically unstable.

The optimal approach is to utilize the QR decomposition of $G$, $G=QR$, where $Q
\in \Lrm(\XX)$ is unitary and $R \in \Lrm(\XX)$ is upper triangular. This
procedure is unique if $G$ is invertible, and we require the diagonal elements of
$R$ to be positive. As typical implementations of the QR algorithm do not
consider this restriction, we must enforce it ourselves. The algorithm is as
follows
\begin{enumerate}
\item Generate a Ginibre matrix $G \in \Lrm(\XX)$, $\mathrm{dim}(\XX) = d$

\item Perform the QR decomposition obtaining $Q$ and $R$.

\item Multiply the $i$\textsuperscript{th} column of $Q$ by $r_{ii}/|r_{ii}|$.
\end{enumerate}
This gives us a Haar distributed random unitary. For a detailed analysis of this
algorithm, see~\cite{mezzadri2006generate}. This procedure can be generalized in
order to obtain a random isometry. The only required change is the dimension of
$G$. We simply start with $G \in \Lrm(\XX, \YY)$, where $\dim(\XX)\geq
\dim(\YY)$.

\subsection{Random mixed quantum states}

In this section, we discuss the properties and methods of generating mixed random
quantum states.

Random mixed states can be generated in one of two equivalent ways. The first
one comes from the partial trace of random pure states. Suppose we have a pure
state $\ket{\psi} \in \XX \otimes \YY$. Then we can obtain a random mixed as
\begin{equation}
\rho = \tr_\YY \ketbra{\psi}{\psi}.
\end{equation}
Note that in the case $\dim(\XX)=\dim(\YY)$ we recover the (flat)
Hilbert-Schmidt distribution on the set of quantum states.

An alternative approach is to start with a Ginibre matrix $G \in \Lrm(\XX,
\YY)$. We obtain a random quantum state $\rho$ as
\begin{equation}
\rho = GG^\dagger/\Tr(GG^\dagger).
\end{equation}
It can be easily verified that this approach is equivalent to the one utilizing
random pure states. First, note that in both cases, we start with $\dim(\XX)
\dim(\YY)$ complex random numbers sampled from the standard normal
distribution. Next, we only need to note that taking the partial trace of a
pure state $\ket{\psi}$ is equivalent to calculating $AA^\dagger$ where $A$ is
a matrix obtained from reshaping $\ket{\psi}$.

The properties of these states have been extensively studied. We will omit
stating all the properties here, and refer the reader
to~\cite{wootters1990random,zyczkowski2001induced,
sommers2004statistical,puchala2016distinguishability,zhang2017average,zhang2017average2}.

\appendix
\end{document}